\documentclass{article}
\usepackage[T1]{fontenc}
\usepackage[utf8]{inputenc}

\date{}

\usepackage{amssymb}

\usepackage{authblk}

\usepackage{tikz,
			pgflibraryshapes}
\usepackage{latexsym,
			optparams
			}
\usepackage{array,
			xspace,
			multirow,
			hhline,
			graphicx,
			xcolor,
			colortbl,
			tabularx,
			booktabs,
			fixltx2e
			}

\usepackage{algorithm,algorithmic}

\usepackage{balance}
			
\usepackage{optparams,xspace}			
\usepackage{ifthen}
\usepackage{comment}
\usepackage{ragged2e}
\usepackage{mathrsfs}
\usepackage{eucal}

\usepackage[round]{natbib}
\newtheorem{theorem}{Theorem}%
\newtheorem{lemma}{Lemma}%
\newtheorem{proposition}{Proposition}%
\newtheorem{corollary}{Corollary}%
\newtheorem{example}{Example}%

\newcommand{\exref}[1]{Example~\ref{#1}}

\newcommand{\ie}{i.e.,\xspace}
\newcommand{\eg}{e.g.,\xspace}

\usepackage{ifthen}
\usepackage{enumerate}
\usepackage{comment}

\usepackage{fixltx2e}
\usepackage{stfloats}
\usepackage{wrapfig}
\usepackage{algorithm}
\usepackage{algorithmic}
\usepackage{booktabs}  
\usepackage{verbatim,ifthen}
\usetikzlibrary{arrows}
\usetikzlibrary{decorations.pathreplacing}
\usetikzlibrary{patterns}

\usetikzlibrary{positioning,chains,fit,shapes,calc}
\usepackage{pgflibraryshapes}
\usepackage{pgf}
\usepackage{paralist}

\usepackage{mathtools}

\usepackage{booktabs}  
\usepackage{accents}  
\usepackage{algorithm,algorithmic} 
\usepackage[inline]{enumitem}
\usepackage{mathrsfs}

\usepackage{tikz}
\usetikzlibrary{arrows}

\newcommand{\midd}{\mathrel{:}}
\newcommand{\sepp}{\mathrel{|}}

\renewcommand{\ie}{i.e., }

\newcommand{\set}[1]{\{#1\}}

\newcommand{\vars}{\ensuremath{\mathit{V}}}

\newlength{\wordlength}
\newcommand{\wordbox}[3][c]{\settowidth{\wordlength}{#3}\makebox[\wordlength][#1]{#2}}
\newcommand{\mathwordbox}[3][c]{\settowidth{\wordlength}{$#3$}\makebox[\wordlength][#1]{$#2$}}

\newcommand{\TRANS}{\ensuremath{\mathbf{P}}}

\newcommand{\sep}{\mathwordbox{|}{x}}

\newcommand{\trans}{\mathit{trans}}
\newcommand{\hatsub}{_{i\vec\jmath\leftarrow \vec b}}
\newcommand{\hatsubb}{_{\vec{\imath\jmath}\leftarrow\vec b}}

\newcommand{\oone}{\overline{1}}
\newcommand{\otwo}{\overline{2}}
\newcommand{\othree}{\overline{3}}
\newcommand{\ofour}{\overline{4}}

\newcommand{\bphnote}[1]{
	\begin{trivlist}
		\item{\color{cyan!40!black}\textbf{BPH:}\hspace{1em}#1}
	\end{trivlist}
}
\newcommand{\jlnote}[1]{
	\begin{trivlist}
		\item{\color{red!80!black}\textbf{JL:}\hspace{1em}#1}
	\end{trivlist}
}

\newcommand{\hanote}[1]{
	\begin{trivlist}
		\item{\color{blue!50!black}\textbf{HA:}\hspace{1em}#1}
	\end{trivlist}
}

\newcommand{\mjwnote}[1]{
	\begin{trivlist}
		\item{\color{cyan}\textbf{MW:}\hspace{1em}#1}
	\end{trivlist}
}

\newcommand{\nocomments}{
	\renewcommand{\hanote}[1]{}
	\renewcommand{\jlnote}[1]{}
	\renewcommand{\bphnote}[1]{}
	\renewcommand{\mjwnote}[1]{}
}

\newcommand{\qed}{\unskip\hspace{1em}\hspace*{\fill}\ensuremath{\Box}}
	\newenvironment{proof}[1][Proof]{\begin{trivlist}%
		\item[\hskip \labelsep {\it #1:}]}{%
   	\qed\end{trivlist}}%

\title{Boolean Hedonic Games\thanks{This paper was presented at the Eleventh Conference on Logic and the Foundations of Game and Decision Theory (LOFT 2014) in Bergen, Norway, July 27-30, 2014.}}
\author[1]{Haris Aziz}
\author[2]{Paul Harrenstein}
\author[3]{J{\'e}r{\^o}me Lang}
\author[2]{Michael Wooldridge}

\affil[1]{\small NICTA and University of New South Wales, Australia}
\affil[2]{\small Department of Computer Science, University of Oxford, UK}
\affil[3]{\small LAMSADE, Universit{\'e} Paris-Dauphine, France}

\begin{document}
	\nocomments
	

	
\sloppy

%

%
%
%

\maketitle

\begin{abstract}
    \noindent We study \emph{hedonic games with dichotomous
      preferences}. Hedonic games are cooperative games in which
    players desire to form coalitions, but only care about the makeup
    of the coalitions of which they are members; they are indifferent
    about the makeup of other coalitions. The assumption of
    dichotomous preferences means that, additionally, each player's
    preference relation partitions the set of coalitions of which that
    player is a member into just two equivalence classes: satisfactory
    and unsatisfactory. A player is indifferent between satisfactory
    coalitions, and is indifferent between unsatisfactory coalitions,
    but strictly prefers any satisfactory coalition over any
    unsatisfactory coalition. We develop a succinct representation for
    such games, in which each player's preference relation is
    represented by a propositional formula.  We show how solution
    concepts for hedonic games with dichotomous preferences are
    characterised by propositional formulas.
\end{abstract}

%
%
%

		\section{Introduction}\label{intro}
		
	\hanote{We use set of players, group of players and coalition of players. Are they terms used systematically for different contexts?}
	\bphnote{My suggestion: S is  a coalition of players if a partition~$\pi$ is given and $S\in\pi$. A group of players I would say is a subset of players that need not be a coalition in a given partition. ``Set of players'' I would only use once, when introducing the framework: ``$N$ is a set of players''.}

	Hedonic games are cooperative games in which players desire to form
	coalitions, but only care about the makeup of the coalitions of which
	they are members; they are indifferent about the makeup of other
	coalitions~\citep{dreze:80a,chalkiadakis:2011a}.  Because the specification of a
	hedonic game requires the expression of each player's ranking over all
	sets of players including him, in general, such a specification
	requires exponential space -- and, when used by a centralised
	mechanism, exponential elicitation time. Such an exponential blow-up
	severely limits the practical applicability of hedonic games, and for
	this reason researchers have investigated compactly represented
	hedonic games. One approach to this problem has been to consider
	possible restrictions on the possible preferences that players
	have. For example, one may assume that each player specifies only a
	ranking over single players, and that her preferences over coalitions
	are defined according to the identity of the best (respectively,
	worst) element of the coalition~\citep{CeHa04a,Cech08a}. One may also
	assume that each player's preferences depend only on the number of
	players in her coalition~\citep{bogomolnaia:2002}. These representations come
	with a domain restriction, i.e., a loss of expressivity:
	\citet{ElWo09a} consider a fully expressive representation for hedonic
	games, based on weighted logical formulas.
	In the worst case, the representation
	of~\citeauthor{ElWo09a} 
	requires space exponential in the number of players,
	but in many cases the space requirement is much smaller.

	In this paper, we consider another natural restriction on player
	preferences. We consider hedonic games with \emph{dichotomous
	  preferences}. The assumption of dichotomous preferences means that
	each player's preference relation partitions the set of coalitions of
	which that player is a member into just two equivalence classes:
	satisfactory and unsatisfactory. A player is indifferent between
	satisfactory coalitions, and is indifferent between unsatisfactory
	coalitions, but strictly prefers any satisfactory coalition over any
	unsatisfactory coalition.
	
	While to the best of our knowledge dichotomous preferences have not
	been previously studied in the context of hedonic games, they have of
	course been studied in other economic settings, such as by
	\citet{BMS05a}, \citet{BoMo04a}, and \citet{Bouveret08Jair} in the
	context of fair division, by ~\citet{HHMW01a} in the context of Boolean
	games, by \citet{KP02} in the context of belief merging, by
	\citet{BoMo04a} in the context of matching, and by \citet{BrFi07c} (and
	many others) in the context of approval voting.

	When the space of all possible alternatives has a combinatorial
	structure, propositional formulas are a very natural representation of
	dichotomous preferences. In such a representation, variables
	correspond to goods (in fair division), outcome variables (Boolean
	games), state variables (belief merging), or players (coalition
	formation).  In the latter case, which we will be concerned with in
	the present paper, each player~$i$ can express her preferences over
	coalitions containing her by using propositional atoms of the form~$ij$ ($j \neq i$), meaning that~$j$ is in the same coalition as
	$i$. Thus, for example, player~1 can express by the formula $(12 \vee
	13) \wedge \neg 14$ that he wants to be in a coalition with player~$2$
	or with player~$3$, but not with player~$4$.
	Our primary aim in this paper is to present such a propositional
	framework for specifying hedonic games and computing various solution
	concepts. We will first define a propositional logic using atoms of
	the form $ij$, together with domain axioms expressing that the output
	of the game should be a partition of the set of players. Then we 
	consider a range of solution concepts, and show that they can be
	characterised by some specific classes of (sometimes polysize)
	formulas, and solved using propositional satisfiability solvers.  The
	result is a simple, natural, and compact representation scheme for
	expressing preferences, and a machinery based on satisfiability for
	computing partitions satisfying some specific stability criteria such
	as Nash stability or core stability.

	\hanote{Need to clearly show why the new work is different from previous logical approaches to coalition formation. Is it that the framework is old but only characterisations are new?}

	\section{Preliminaries}

	In this section, we  recall some definitions relating to
	coalitions, coalition structures (or partitions), 
	and hedonic games. See, e.g.,~\cite{chalkiadakis:2011a} for an in-depth
	discussion of these and related concepts.

	\paragraph{Coalitions and Partitions}

	\bphnote{Standard notation in the context of hedonic games is: $S,T,\dots$ for coalitions,~$\mathscr N_i$ for the set of coalitions~$i$ belongs to. Let's stick to that, even though~$\mathscr N_i$ is rather wild and nevertheless used quite a lot in what follows.}

	We consider a setting in which there is a set~$N$ of~$n$ players with
	typical elements~$i,j,k,\dots$. Players can form \emph{coalitions},
	which we will denote by~$S,T,\dots$. A coalition is simply a subset of
	the players~$N$. One may usefully think of the players as getting together to form
	teams that will work together.  A \emph{coalition structure} is an
	exhaustive partition~$\pi=\set{S_1,\dots,S_m}$ of the players into
	disjoint coalitions, \ie $S_1\cup\dots\cup S_m=N$ and $S_i\cap
	S_j=\emptyset$ for all~$S_i,S_j\in\pi$ such that~$i \neq j$.  {For
	  technical convenience, we slightly deviate from standard conventions
	  and require that every coalition structure~$\pi$ contains the empty
	  set~$\emptyset$.}
	We commonly refer to coalition structures simply as \emph{partitions}.
	In examples, we also write, \eg
	$[\mathwordbox{12}{mi}|\mathwordbox{34}{mi}|\mathwordbox{5}{m}]$
	rather than the more cumbersome
	$\set{\set{1,2},\set{3,4},\set{5},\emptyset}$. 
	For each player~$i$ in~$N$,
	we let $\mathscr N_i=\set{S\subseteq N\midd i\in S}$ denote the set of
	coalitions over~$N$ that contain~$i$.  If $\pi=\set{S_1,\dots,S_m}$ is
	a partition, then~$\pi(i)$ refers to the coalition in~$\pi$ that
	player~$i$ is a member of.

	The notion of players leaving their own coalition and joining another
	lies at the basis of many of the solution concepts that we will come
	to consider. We introduce some notation to represent such situations.
	For~$T$ a group of players (not necessarily a coalition in~$\pi$), by~$\pi|_T$ we
	refer to the partition $\set{S_1\cap T,\dots,S_m\cap T}$ and we
	write~$\pi|_{-T}$ for~$\pi|_{N\setminus T}$. Moreover, for~$S$ a
	coalition in partition~$\pi|_{-T}$, we use $\pi[T \to S]$ to refer to
	the partition that results if the players in~$T$ leave their
	respective coalitions in~$\pi$ and join coalition~$S$. We also
	allow~$T$ to form a coalition of its own, in which case we write $\pi[T\to\emptyset]$.
	Formally, we have, for~$S\in\pi|_{-T}$, 
	\begin{align*}
		\pi[T\to S]	& \mathwordbox[c]{=}{==} \set{S_j\in\pi|_{-T}\midd S_j\neq S}\cup\set{S\cup T,\emptyset}\text.
	\end{align*}
	\bphnote{In a later version we might also take into account \emph{simultaneous} deviations of coalitions. This would require slightly more complicated definitions, which are undesirable here.}
	If~$T$ is a singleton~$\set i$ we also write~$\pi|_{-i}$ and $\pi[i\to S]$ instead of~$\pi|_{-\set i}$ and $\pi[\set i\to S]$, respectively.
	Thus, \eg $S\cup\set i\in\pi[i\to S]$ and $\pi[i\to\pi(i)\setminus\set i]=\pi$. 

	Finally, define $\pi[i\rightleftarrows j]$ as the partition where $i$ and $j$ exchange their places, i.e.:
	\[
	\scalebox{1}[1]{$
	\pi[i\rightleftarrows j] = (\pi\setminus\set{\pi(i),\pi(j)})\cup\set{(\pi(i)\setminus\set i)\cup\set j,(\pi(j)\setminus\set j)\cup\set i}\text.
	$}
	\]
	Thus, for partition $\pi = [123|45]$, we have $\pi(1) = \pi(2)=\{1,2,3\}$ and $\pi(4)=\set{4,5}$. Furthermore, $\pi|_{\set{1,2,4,5}}=[12|45]$ and $\pi|_{-\set{3,4}}=[12|5]$. Also, $\pi[1\to \{4,5\}] = [23|145]$, $\pi[1\to \emptyset] = [1|23|45]$, and $\pi[3\rightleftarrows 4] = [124|35]$.

	%

	\noindent\paragraph{Hedonic games}
	Hedonic games are the class of coalition formation games in which each
	player is only interested in the coalition he is a member of, and is
	indifferent as to how the players outside his own coalition are
	grouped. Hedonic games were originally introduced by \cite{dreze:80a} and further developed by, e.g., \cite{bogomolnaia:2002}. Also see \cite{hajdukova:2006} for a survey from a more computational point of view. Formally, a \emph{hedonic game} is a
	tuple~$(N,R_1,\dots,R_n)$, where~$R_i$ represents~$i$'s transitive,
	reflexive, and complete preferences over the set of all coalitions
	$\mathscr N_i$ containing $i$. Thus, $S \mathrel{R_i}T$ intuitively
	signifies that player~$i$ considers coalition~$S$ at least as
	desirable as coalition~$T$, where~$S$ and~$T$ are coalitions
	in~$\mathscr N_i$. By~$P_i$ and~$I_i$ we denote the strict and the
	indifferent part of~$R_i$, respectively. The preferences~$R_i$ of a
	player~$i$ are said to be \emph{dichotomous} whenever $\mathscr N_i$
	can be partitioned into two disjoint sets~$\mathscr N_i^+$
	and~$\mathscr N_i^-$ such that~$i$ strictly prefers all coalitions
	in~$\mathscr N_i^+$ to those in~$\mathscr N_i^-$ and is indifferent
	otherwise, \ie $S\mathrel{P_i}T$ if and only if $S\in \mathscr N_i^+$
	and $T\in \mathscr N_i^-$. 	A coalition~$S$ in~$\mathscr N_i$ is \emph{acceptable} to~$i$ if~$i$ (weakly) prefers~$S$ to coalition~$\{i\}$, where he is on his own, \ie if~$S\mathrel{R_i}\set i$. By contrast, we say that a coalition~$S$ is \emph{satisfactory} or \emph{desirable} for~$i$ if $S\in\mathscr N_i^+$. Satisfactory partitions are thus generally acceptable to all players. The implication in the other direction, however, does not hold.

	We lift preferences on coalitions to preferences on partitions in a natural way: player~$i$ prefers partition~$\pi$ to partition~$\pi'$ whenever~$i$ prefers coalition~$\pi(i)$ to coalition~$\pi'(i)$. We also  extend the concepts of acceptability and desirability of coalitions  to partitions.

	\begin{example}\label{ex1}
		Consider the following Boolean game with four players, $1$, $2$, $3$, and $4$, whose (dichotomous) preferences are as follows. (Indifferences are indicated by commas.)
		\begin{align*}
			1\colon	& \scalebox{1}[1]{$\set{1,2,3},\set{1,2,4},\set{1,3,4},\set{1,2,3,4}$}	\mathrel P_1 \scalebox{1}[1]{$\set{1},\set{1,2},\set{1,3},\set{1,4}$}\\
			2\colon	& \scalebox{1}[1]{$\set{2,1,3},\set{2,1,4},\set{2,3,4}$}					\mathrel P_2 \scalebox{1}[1]{$\set{2},\set{2,1},\set{2,3},\set{2,4},\set{2,1,3,4}$}\\
			3\colon	& \scalebox{1}[1]{$\set{3,1},\set{3,2},\set{3,1,2}$}						\mathrel P_3 \scalebox{1}[1]{$\set{3},\set{3,4},\set{3,1,4},\set{3,2,4},\set{3,1,2,4}$}\\
			4\colon	& \scalebox{1}[1]{$\set{4,1},\set{4,2},\set{4,3},\set{4,1,2},\set{4,1,3},\set 4$}					\mathrel P_4 \scalebox{1}[1]{$\set{4,2,3},\set{4,1,2,3}$}
		\end{align*}
		Thus, player~$1$ wants to be in a coalition of at least three and player~$2$ wishes to be in a coalition of exactly three. 
		Moreover, player~$3$ wants to be in the same coalition as player~$1$ or as~$2$. He does not want to be in a coalition with player~$4$.
		Finally, player~$4$ does not want to be with players~$2$ and~$3$ together. There is exactly one partition that is satisfactory for all four
		players, namely $[123\sep 4]$. For players~$1$,~$2$, and~$3$, all coalitions are acceptable. For player~$4$, however, $\set{4,2,3}$ and $\set{1,2,3,4}$ are unacceptable.
	\end{example}

	\paragraph{Solution Concepts for Hedonic Games}

	A \emph{solution concept} associates with every hedonic game~$(N,R_1,\dots,R_n)$ a (possibly empty) set of partitions of~$N$. Here we review some of the most common solution concepts for hedonic games.

	\label{page:stability_concepts}
	\begin{itemize}[label={},leftmargin=0em,itemsep=0ex]
	\item Individual rationality captures the idea that every player prefers the coalition he is in to being on his own, i.e., that coalitions are acceptable to its members. Thus, formally,~$\pi$ is \emph{individually rational} if, for all players~$i$ in~$N$,
			\[ 
				\text{$\pi(i)\mathrel{R_i}\{i\}$.}
			\]
	This condition is obviously equivalent to $\pi\mathrel{R_i}\pi[i\to\emptyset]$.
	\item For dichotomous hedonic games, a partition~$\pi$ is said to be \emph{social welfare optimal} if it maximises the number of players who are in a satisfactory coalition, that is, if~$\pi$ maximises $|\set{i\in N\midd \pi(i)\in\mathscr N^+_i}|$. In a similar way, a partition~$\pi$ is \emph{Pareto optimal} if it maximises the set of players being in a satisfactory coalition with respect to set-inclusion, that is, if there is no partition~$\pi'$ with
	\[
		\set{i\in N\midd \pi(i)\in\mathscr N^+_i}\subsetneq
		\set{i\in N\midd \pi'(i)\in\mathscr N^+_i}\text.
	\]
	 In the extreme case in which every player is in a most preferred coalition,~$\pi$ is said to be \emph{perfect} \citep[cf.,~][]{ABH11c}.
		A perfect partition satisfies any other of our stability concepts.
	
	\item A partition is \emph{Nash stable} if no player would like to
	  unilaterally abandon the coalition he is in and join any other
	  existing coalition or stay on his own, that is, if, for all~$i\in N$ and all $S\in\pi$,
			\[
				\text{$\pi(i)\mathrel{R_i}S \cup \{i\}$.
			}
			\]
	Observe that this condition is equivalent to $\pi\mathrel{R_i}\pi[i\to S]$.		
	\item Core stability concepts consider group deviations instead of individual ones. A group of players, possibly from different coalitions, is said to block a partition if they would all benefit by joining together in a separate coalition. Formally,~$T$ \emph{blocks} (or \emph{is blocking}) partition~$\pi$ if, for all~$i\in T$,
	\[
		T \mathrel{P_i} \pi(i)\text.
	\]
	Thus,~$T$ blocks~$\pi$ if and only if $\pi[T\to\emptyset]\mathrel{P_i}\pi$ for all $i\in T$.
	 A group~$T$ \emph{weakly blocks} (or \emph{is weakly blocking})~$\pi$ if 	
	$T \mathrel{R_i} \pi(i)$ holds for all~$i\in T$ and $T \mathrel{P_i} \pi(i)$ holds for some $i \in T$.
	Then, $\pi$ is \emph{core stable} if no group is blocking it
	and $\pi$ is \emph{strict core stable} if no group is weakly blocking it.
	\item Partition $\pi$ is {envy-free} if no player is envious of another player, that is, if no player~$i$ would prefer to change places with another player~$j$. Formally, partition~$\pi$ is \emph{envy-free} if, for all players~$i$ and~$j$,
	\[
		\pi\mathrel{R_i}\pi[i\leftrightarrows j].
	\]
	If $\pi[i\leftrightarrows j]\mathrel{P_i}\pi$ we also say that player~$i$ \emph{envies} player~$j$.
	\end{itemize}
	
{
\setcounter{example}{0}
	\begin{example}[continued] In our example, in partition~$[1,2,3\sep 4]$ each player is in a most preferred coalition. As such~$[1,2,3\sep 4]$ is perfect as well as social welfare optimal and satisfies all solution concepts mentioned above. 
	Moreover, all partitions except $[1\sep 2,3,4]$ and $[1,2,3,4]$ individually rational. 
	
	Now, consider partition $\pi=[1\sep 2,3\sep 4]$.
	Here, player~$2$ does not want to abandon her coalition~$\set{2,3}$ and join another as she prefers none of the following partitions to~$\pi$: $\pi[2\to\set 1]=[1,2\sep 3\sep 4]$, $\pi[2\to\set{2,3}]=[1\sep 2,3\sep 4]$, $\pi[2\to\set 4]$, and $\pi[2\to\emptyset]=[1\sep 2\sep 3\sep 4]$.
	As, however, $\pi[1\to\set{2,3}]=[1,2,3\sep 4]$ and $[1,2,3\sep 4]\mathrel P_1\pi$, partition~$\pi$ is not Nash stable.

	Also observe that for $\pi=[1\sep 2,3\sep 4]$ the group $\set{1,2,3}$ is strongly blocking, as $\pi[\set{1,2,4}\to\emptyset]=[1,2,4\sep 3]$ and $[1,2,4\sep 3]\mathrel P_{i}\pi$ for all $i\in\set{1,2,4}$.
	Thus, $\pi$ is not core stable. By contrast, $[1,4\sep 2,3]$ is core stable as only player~$1$ and~$2$ are not satisfied and both of them will only be if they can form a blocking coalition of exactly three.
	However, $\set{1,2,4}$ is still weakly blocking, and as such $[1,4\sep 2,3]$ is not strict core stable.
					
				For envy-freeness, consider partition~$\pi'=[1\sep 2,4\sep 3]$. Then, player~$3$ envies player~$4$, as $\pi'[3\leftrightarrows 4]=[1\sep 2,3\sep 4]$ and $[1\sep 2,3\sep 4]\mathrel{P_3}\pi'$. By contrast, player~$3$ does not envy player~$2$: we have $\pi'[3\leftrightarrows 2]=[1\sep 2\sep 3,4]$ but not $[1\sep 2\sep 3,4]\mathrel{P_3}\pi'$.
	\end{example}
}

	\mjwnote{
	\paragraph{Strategyproofness} Strategyproofness promises to be very interesting as in hedonic Boolean games and Boolean coalition formation games you can reason about preferences logically.
	}
		\hanote{I think SP can be easily achieved by translating the setting to a coalition choice setting in which partitions are the alternatives. In that case approval voting is strategyproof for dichotomous preferences. The problem of finding a partition with the maximum number of approvals is at least hard as checking whether a perfect partition exists. Hence the problem is NP-hard.}
\bphnote{This is a very interesting remark and we could insert it somewhere in the text.}


	\section{{A Logic for Coalition Structures}}

	In this section, we develop a logic for representing coalition
	structures. We will then use this logic as a compact specification
	language for dichotomous preference relations in hedonic games.

	\paragraph{Syntax}
	Given a set~$N$ of~$n$ players, we define a propositional language~$L_N$
	built from the usual connectives and with for every (unordered)
	pair~$\set{i,j}$ of distinct players a propositional variable
	$p_{\set{i,j}}$. The set of propositional variables we denote
	by~$\vars$. Observe that $|V|=\binom{n}{2}$.  \bphnote{I am not overly happy with the notation
	  $\vars$ for propositional variables/symbols. I would
	  prefer~$X$, $A$, $Q$, or perhaps even $\Phi$. The reason is that
	  $\exists \vars\varphi$ looks awkward. I am no great fan of
	  introducing the notation $\exists i\varphi$ for
	  $\exists\vars_i\varphi$ either as this may be confused with
	  $\bigvee_i\varphi$, which we are also using. Moreover, if we choose
	  for~$X$, we can use~$x,y,z,\dots$ as metavariables over
	  propositional variables and $\vec x$ for sequences/vectors of
	  propositional variables. If you agree, please, change the macro
	  \texttt{$\backslash$vars} in the preamble accordingly.}  For
	notational convenience we will write~$ij$ for~$p_{\set{i,j}}$. Thus,~$ij$ and~$ji$ refer to the same symbol. The
	language is interpreted on coalition structures on~$N$ and the
	informal meaning of~$ij$ is ``$i$ and $j$ are in the same coalition''.
	Formally, the formulas of the language~$L_N$, with typical
	element~$\varphi$ is given by the following grammar
	\[
		\varphi	\mathrel{\Coloneqq}	ij
				\sepp	\neg\varphi
				\sepp	(\varphi\vee\varphi)
	\]
	where $i,j\in N$ and $i\neq j$. By $|\varphi|$ we denote the \emph{size} of~$\varphi$.

	For a given coalition~$S$ of players, we write $\vars_S$ for the propositional variables in which some~$i\in S$ appears, i.e.,
	\[
		\vars_S	= \set{ij\in\vars\midd \text{$i\in S$ or $j\in S$}}\text.
	\]
	Note that for distinct players~$i$ and~$j$ we have~$\vars_i\cap\vars_j=\set{ij}$.
	The propositional language over~$\vars_S$ we denote by $L_S$. We write~$\vars_i$ and~$L_i$ for $\vars_{\set i}$ and $L_{\set i}$, respectively. 
	The remaining classical connectives~$\bot$, $\top$, $\wedge$, $\to$,
	and $\leftrightarrow$ are defined in the usual way. Moreover, for
	formulas ${\psi_1,\dots,\psi_k}$ of formulas, we have
	$\bigwedge_{1\le m\le k}\psi_m$ and $\bigvee_{1\le m\le k}\psi_m$ abbreviate
	$\psi_1\wedge\dots\wedge\psi_k$ and $\psi_1\vee\dots\vee\psi_k$,
	respectively.
%
%
	We also make use of the following useful
	notational shorthand:
		\begin{align*}
			i_1 \cdots i_m \overline{i}_{m+1} \cdots \overline{i}_p
				&	=	
			\bigwedge_{1 \leq j \leq m}i_1 i_j\wedge \bigwedge_{m < k \leq p}\neg i_1 i_k\text.
		\end{align*}
	\bphnote{I have changed the definition here, in such a way that $1234=12\wedge13\wedge14$. This is not only shorter, but also helps to guarantee that our examples are still examples of hedonic games. The examples, however, need to be reformulated in such a way that when formulating a player's goal that player should always be mentioned first. $1234$ is \emph{not} the same symbol as, for instance, $3214$. Another option, to get things correct is not requiring that in hedonic games the goals of each player~$i$ is actually from~$L_i$ but \emph{equivalent} to some formula in~$L_i$. This might be more elegant.}
	Thus, 
	$i_1 \cdots i_m \overline{i}_{m+1} \cdots \overline{i}_p$ conveys that  $i_1, \dots, i_m$ are in the same coalition and each of them in another coalition than ${i}_{m+1} \cdots {i}_p$. Thus, where $N=\set{1,2,3,4}$,
	$
		12\othree\ofour \vee 13\otwo\ofour \vee 14\otwo\othree 
	$
abbreviates
	$
		(12 \wedge \neg 13\wedge \neg 14) \vee (13 \wedge \neg 12\wedge\neg 14) \vee(14\wedge\neg 12\wedge \neg 13)
	$
	and signifies that player~$1$ is in a coalition of two players.
	

		\paragraph{Semantics}
		We interpret the formulas of $L_N$ on partitions~$\pi$ as follows. 
		\[
			\begin{array}{lcl}
				\pi \models ij& \mbox{if and only if} & \pi(i)=\pi(j)\\
				\pi \models \neg\varphi & \mbox{if and only if} & \pi\not\models\varphi\\
				\pi \models \varphi\to\psi & \mbox{if and only if} & \mbox{$\pi\not\models\varphi$ or $\pi\models\psi$}
			\end{array}
		\]
		For~$\Psi\subseteq L_N$, we have $\Psi\models\varphi$ if $\pi\models\psi$ for all $\psi\in\Psi$ implies $\pi\models\varphi$. If $\Psi=\emptyset$, we write $\models\varphi$ and say that~$\varphi$ is \emph{valid}.
		
	  Notice that partitions play a dual role in our framework: both
	  their initial role as coalition structures, and the role of models
	  in our logic. This dual role is key to using formulas of our
	  propositional language as a specification language for preference
	  relations. Thus, e.g., partition~$[1|2|345]$ satisfies 
the following formulas of $L_N$: $345$, $3\oone$,  $345\overline{1}\overline{2}$, $\neg 12 \wedge (23 \vee 34)$, and $12 \leftrightarrow 23$.


	\paragraph{Axiomatisation} We have the following axiom schemes
	    for mutually distinct players~$i$,~$j$, and~$k$,
		\begin{enumerate}[label=\ensuremath{(\mathrm A\arabic*)},leftmargin=3em]
		\setcounter{enumi}{-1}
			\item\label{axiom:taut} all propositional tautologies
			\item\label{axiom:trans} $ij\wedge jk \to ik$ \hfill (\emph{transitivity})
		\end{enumerate}
		as well as \emph{modus ponens} as the only rule of the system:
		\begin{enumerate}[label=\ensuremath{(\mathrm{MP})},leftmargin=3em]
		\item \text{from \wordbox{$\varphi$}{ $\varphi$ } and \wordbox{$\varphi\to\psi$}{ $\varphi\to\psi$ } infer \wordbox{$\psi$.}{ $\psi.$ }}
		\hfill (\text{\emph{modus ponens}})
		\end{enumerate}
		The resulting logic we refer to as \TRANS\ and write $\Psi\vdash_{\TRANS}\varphi$ if there is a derivation of~$\varphi$ from~$\Psi$, \ref{axiom:taut}, and \ref{axiom:trans},  using modus ponens.

		\begin{theorem}[Completeness]\label{thm:completeness} Let $\Psi\cup\set\varphi\subseteq L_N$. Then,
		\[
			\text{$\Psi\vdash_\TRANS\varphi$ \wordbox{ if and only if }{ iif and only iff } $\Psi\models\varphi$.}
		\]
		\end{theorem}
		\begin{proof}[(sketch)]
			Soundness is straightforward. For completeness a standard Lindenbaum construction can be used. To this end, assume $\Psi\not\vdash_{\TRANS}\varphi$. Then, $\Psi\cup\set{\neg\varphi}$ is consistent and can as such be extended to a maximal consistent theory $\Psi^*$. Define a relation~$\sim_{\Psi^*}$ such that for all $i,j\in N$, 
			\[
				\text{$i\sim_{\Psi^*}j$ \wordbox{ if and only if }{ iif and only iff } $ij\in\Psi^*$.}
			\]
		The axiom schemes~\ref{axiom:taut} and~\ref{axiom:trans} ensure that~$\sim_{\Psi^*}$ is a well-defined equivalence relation. Let $[\mathwordbox{i}{n}]_{\sim_{\Psi^*}}=\set{j\in N\midd i\sim_{\Psi^*}j}$ be the equivalence class under~$\sim_{\Psi^*}$ to which player~$i$ belongs. Then define the partition $\pi_{\Psi^*}=\set{[\mathwordbox{i}{n}]_{\sim_{\Psi^*}}\midd i\in N}$. By a straightforward structural induction, it can then be shown that for all $\psi\in L_N$,
		\[
			\text{$\pi_{\Psi^*}\models\psi$	\wordbox{ if and only if }{ iif and only iff } $\psi\in\Psi^*$.}
		\]
	It follows that $\pi_{\Psi^*}\models\Psi$ and $\pi_{\Psi^*}\not\models\varphi$. Hence, $\Psi\not\models\varphi$.
		\end{proof}

	Alternatively, one can reason with coalition structures in standard
	propositional logic, by writing the transitivity axiom directly as a
	propositional logic formula. Let
	\[
		\trans = \bigwedge_{i,j,k\in N} (ij \wedge jk \rightarrow ik)\text.
	\]
	Then, for any propositional formulas~$\varphi$ and~$\psi$ of $L_N$,
	\[
		\varphi \vdash_\TRANS \psi \text{\wordbox{ if and only if }{ iif and only iff }} \varphi \wedge \trans \vdash \psi
		\]
	that is, checking whether a formula~$\varphi$ implies another formula~$\psi$ in~$\TRANS$ is equivalent to saying that~$\varphi$ together with the transitivity constraint implies~$\psi$. This means that reasoning tasks in~$\TRANS$ can be done with a classical propositional theorem prover. In what follows we say that two formulas~$\varphi$ and~$\psi$ are $\TRANS$-equivalent whenever their equivalence can be proven in~$\TRANS$, \ie $\vdash_\TRANS\varphi\leftrightarrow\psi$.


	\section{Boolean Hedonic Games}
	The denotation of a formula $\varphi$ of our propositional language is
	a set of coalition structures, and we can naturally interpret these as
	being the desirable or satisfactory coalition structures for a particular player.
	Thus, instead of writing a hedonic game with dichotomous preferences
	as a structure $(N,R_1, \ldots, R_n)$, in which we explicitly
	enumerate preference relations $R_i$, we can instead write
	$(N,\gamma_1,\dots,\gamma_n)$, where $\gamma_i$ is a formula of our
	propositional language that acts as a specification of the preference
	relation $R_i$.  Intuitively,~$\gamma_i$ represents player~$i$'s
	`goal' and player~$i$ is satisfied if his goal is achieved and
	unsatisfied if he is not. We refer to a structure $(N,\gamma_1,
	\ldots, \gamma_n)$ as a \emph{Boolean hedonic game}.  Thus, a Boolean
	hedonic game $(N,\gamma_1,\dots,\gamma_{n})$ represents the
	(standard) hedonic game $(N,R_1,\dots,R_n)$ with for each~$i$,
		\[
			\text{$\pi(i)\mathrel{R_i}\pi'(i)$ \wordbox{if and only if}{xif and only ifx} $\pi\models\gamma_i$ implies $\pi'\models\gamma_i$.}
		\]
	  Observe that, defined thus, the preferences of each player in a
	  hedonic Boolean game are dichotomous.  
  
	  It should be clear that every dichotomous preference relation~$R_i$ can be
	  specified by a propositional formula~$\gamma_i$, and hence our
	  propositional language forms a fully expressive representation
	  scheme for Boolean hedonic games.%
	  \footnote{Let~$i$ be a player with dichotomous preferences~$R_i$ and let~$X_i$ be the set of coalitions 
	  most preferred by~$i$, i.e., 
	  $S \in X_i$ if and only if $S \mathrel{R_i}S'$ for all coalitions~$S$ and~$S'$ containing~$i$.
	Then,~$R_i$ is represented by following formula of~$L_i$ in disjunctive normal form:
	 \[
	  	\bigvee_{S  \in X_i}\Big(\bigwedge_{j\in S}ij\wedge\bigwedge_{k\notin S}\neg ik\Big)\text.
	 \]
	   } In fact, formulas in~$L_N$ are strictly more expressive in the sense that they can represent \emph{any} dichotomous preference relation over partitions rather than just preference relations over partitions as induced by a preference relation~$R_i$ for a player~$i$ over \emph{coalitions} in~$\mathscr N_i$. We find, however, that every Boolean hedonic game~$(N,\gamma_1,\dots,\gamma_n)$ represents a hedonic game with dichotomous preferences provided that every player's goal~$\gamma_i$ is equivalent to a formula in the language~$L_i$, the sublanguage of~$L_N$ in which only variables in $V_{i}=\set{ij\midd j\in N\setminus\set i}$ occur. Intuitively, formulas in~$L_i$ only convey information about the coalitions player~$i$ is in or she is not in. 
	\begin{proposition}\label{proposition:characterisation_hedonic_aspect}
	If a Boolean hedonic game $(N,\gamma_1,\dots,\gamma_n)$ represents a hedonic game with dichotomous preferences, then for every player~$i$ there is a formula~$\varphi_i\in L_i$ that is $\TRANS$-equivalent to~$\gamma_i$. Moreover, if for every player~$i$ there is a formula~$\varphi_i\in L_i$ that is $\TRANS$-equivalent to~$\gamma_i$, then $(N,\gamma_1,\dots,\gamma_n)$ represents a hedonic game with dichotomous preferences.
	\end{proposition}

	\begin{proof}[(sketch)]
		For a player~$i$ and $\varphi$ a formula in~$L_i$, a straightforward inductive argument shows that 
		\[
		\text{$\pi\models\varphi$ \wordbox{if and only if}{iiif and only ifff} $\pi'\models\varphi$ for all $\pi'$ with $\pi'(i)=\pi(i)$.}
		\]
		Then, the result follows as a corollary.
	\end{proof}

	  Often, the use of propositional
	  formulas $\gamma_i$ gives a `concise' representation of the
	  preference relation $R_i$, although of course in the worst case the
	  shortest formula $\gamma_i$ representing $R_i$ may be of size
	  exponential in the number of players. In what follows, we will write
	  $(N,\gamma_1, \ldots, \gamma_n)$, understanding that we are
	  referring to the game $(N,R_1, \ldots, R_n)$ corresponding to this
	  specification.

{
\setcounter{example}{0}
	\begin{example}[continued]
			The hedonic game with dichotomous preferences in Example~\ref{ex1} is represented by the Boolean hedonic game $(N,\gamma_1,\gamma_2,\gamma_3,\gamma_4)$ with $N=\set{1,2,3,4}$ and the players' goals given by:
			\begin{align*}
				\gamma_1 &	= 	(123 \vee 124 \vee 134)						&
				\gamma_2 &	= 	(213\ofour \vee 214\othree \vee 234\oone)		\\
				\gamma_3 &	=	(31\vee 32)\wedge\neg34						&
				\gamma_4 &	= 	\neg 423\text.
			\end{align*}
For each player~$i$ we then have that $\pi\models\gamma_i$ if and only if $\pi\in\mathscr N_i^+$.
	\end{example}
	}

	\bphnote{Would this be a good place to present the (easy) characterisation of hedonic games? ($(N,\gamma_1,\dots,\gamma_n)$ is a hedonic game if and only if for each player~$i$ there is a formula $\gamma_i$ in the language~$L_i$ on $\vars_i$ such that $\gamma_i\equiv_{P} \gamma_i$)}

	\section{Substitution and Deviation}\label{section:substition_and_deviation}
	\bphnote{It might be better to have this section \emph{after} Boolean hedonic games are being introduced.}

	We establish a formal link between substitution in formulas of our language and
	the possibility of players deviating from their respective coalition in a
	given partition and joining other coalitions.


\paragraph{Substitution}
We first introduce some formal notation and terminology with respect to substitution of formulas for variables in our logic. 	

For~$ij$ a propositional variable in~$\vars_N$
	and~$\varphi$ and~$\psi$  formulas of~$L_N$, we denote by $\varphi_{ij \leftarrow \psi}$ the \emph{uniform
	substitution} of variable~$ij$ by~$\psi$ in $\varphi$. If 
	$\vec{\imath\jmath}=i_1j_1,\dots,i_kj_k$ is a sequence of~$k$ distinct variables in~$\vars$ 
	and $\vec\psi=\psi_1,\dots,\psi_k$ a sequence of~$k$ formulas,
	\begin{align*}
		\varphi_{\vec{\imath\jmath}\leftarrow\vec\psi}
		&	\mathwordbox{=}{===}  \varphi_{i_1j_1,\dots,i_kj_k\leftarrow \psi_1,\dots,\psi_k}
	\end{align*}
	denotes the \emph{simultaneous substitution} of each $i_mj_m$ by $\psi_m$ ($1\le m\le k$).
	Thus, e.g., $(ij\vee\neg jk)_{ij,jk\leftarrow jk,ik}=jk\vee\neg ik$.
	A special case, which recurs frequently in what follows, is if every~$\psi_i$ is a Boolean, i.e., if $\psi_1,\dots,\psi_k\in\set{\top,\bot}$.
	Sequences $\vec b=b_1,\dots,b_k$ where $b_1,\dots,b_k\in\set{\top,\bot}$ we will also refer to as \emph{Boolean vectors of length~$k$}.
	Thus, e.g., $\top,\bot$ is a Boolean vector of length~$2$ and $(ij\wedge jk\to ki)_{ij,ki\leftarrow\top,\bot}=\top\wedge jk\to\bot$.

	\paragraph{Characterising individual deviations}\label{subsection:characterising_individual_diviations}

	Some of the stability concepts for Boolean hedonic games we consider
	in this paper, e.g., Nash stability, are based on which coalitions an individual player~$i$ can join given a partition~$\pi$. Recall that these coalitions are given by $\pi|_{-i}$. Of course, not all groups of agents are included in~$\pi_{-i}$. For instance, let partition~$\pi$ be given by $[\,12\sep34\sep 5\,]$. Then, player~$1$ can join coalition~$\set{3,4}$ but cannot form a coalition with players~$4$ and~$5$ by unilaterally deviating from~$\pi$. We find that the set~$\pi|_{-i}$ can be characterised in our logic. This furthermore yields a logical characterisation of when a player~$i$ can unilaterally break loose from his coalition, join another one and thereby guarantee that a given formula~$\varphi$ will be satisfied. A particularly interesting case is if~$\varphi$ implies the respective player's goal. We thus gain expressive power with respect to whether a player can \emph{beneficially} deviate from a given partition, a crucial concept.

	\bphnote{The proofs of the following lemmas are included. Presumably, it would be better to leave them out from the submission.}

	\begin{lemma}\label{lemma:forgetting1}
		Let~$\pi$ be a partition,~$i$ a player, $B$~a group of players in $N\setminus\set i$. Let furthermore~$\vec b=b_1,\dots,b_{n-1}$ be a Boolean vector of length~$n-1$ and $i\vec\jmath=ij_1,\dots,ij_{n-1}$ an enumeration of~$\vars_i$ such that $B=\set{j\midd ij\hatsub =\top}$. 
Then,
		\begin{enumerate}[label={\ensuremath{(\roman*)}},leftmargin=2.25em]
			\item\label{item:forgettting1i} $B\in\pi|_{-i}$ 									\wordbox{iff}{iifff} 	$\pi\models\trans\hatsub $,
			\item\label{item:forgettting1ii} $B\in\pi|_{-i}$ and $\pi[i\to B]\models\varphi$ 	\wordbox{iff}{iifff} 	$\pi\models(\varphi\wedge\trans)\hatsub$.
		\end{enumerate}
	\end{lemma}

	\begin{proof}
	{
	\renewcommand{\hatsub}{'}
	We prove~\ref{item:forgettting1i}; the proof for~\ref{item:forgettting1ii} is by structural induction on~$\varphi$ and relies on similar principles as~\ref{item:forgettting1i}.
	As~$\vec b$ and $i\vec\jmath$ are fixed throughout the proof, 
	for better readability, we write~$\varphi'$ for~$\varphi_{i\vec\jmath\leftarrow\vec b}$.

	For the ``only if''-direction, assume that $B\in \pi_{-i}$ as well as $\pi\not\models\trans\hatsub$. Observe that
	$
		\trans\hatsub \mathwordbox{=}{==}
		\bigwedge_{k,l,m}\big(kl\hatsub \wedge lm\hatsub \rightarrow km\hatsub\big)\text.
	$ 
	Accordingly, there are some (mutually distinct)~$k$, $l$,~and~$m$ such that $\pi\not\models kl\hatsub \wedge lm\hatsub \rightarrow km\hatsub$.
	It suffices to consider the following three cases. 
	\begin{align*}
	(a)&\quad i\notin\set{k,l,m}\text,& 
	(b)&\quad i=k\text{,} &
	(c)&\quad i=l\text. 
	\end{align*}

	Case~$(a)$ cannot occur as we would have $kl\hatsub=kl$, $lm\hatsub=lm$, $km\hatsub=km$, and $kl \wedge lm \rightarrow km$ is a theorem of the system. 

	If~$(b)$, then
	$
		\pi\not\models il\hatsub \wedge lm\hatsub \rightarrow im\hatsub\text.
	$
	It follows that $\pi\models il\hatsub$, $\pi\models lm\hatsub$, and $\pi\not\models im\hatsub$. Observe that in this case $lm\hatsub=lm$. Hence, $\pi(l)=\pi(m)$. Also notice that $il\hatsub,im\hatsub\in\set{\top,\bot}$ and, thus, $im\hatsub=\bot$ and $il\hatsub=\top$.
	Accordingly, $l\in B$ but $m\notin B$. As $i\neq m$ and having assumed $B\in\pi|_{-i}$, a contradiction follows:
	\[
		\pi(m)\neq\pi(i)=\pi(l)=\pi(m)\text.
	\]

	If~$(c)$, we have
	$
		\pi\not\models ik\hatsub \wedge im\hatsub \rightarrow km\hatsub\text.
	$
	Thus, $\pi\models ik\hatsub$, $\pi\models im\hatsub$, and $\pi\not\models km\hatsub$. Observe that $km\hatsub=km$. Hence, $\pi(k)\neq\pi(m)$. Moreover, $ik\hatsub,im\hatsub\in\set{\top,\bot}$, from which follows that
	 $ik\hatsub=\top$ and $im\hatsub=\top$.
	Accordingly, both $k,m\in B$. With $B\in\pi|_{-i}$,	 we obtain that $\pi(k)=\pi(m)$, a contradiction.

	For the ``if''-direction, assume $B\notin\pi|_{-i}$ and~$B\neq\emptyset$. Because of the latter, there is some $j\in B$. Accordingly, $ij\hatsub=\top$.
	As $B\notin\pi|_{-i}$, and thus in particular $B\neq \pi(j)\setminus\set{i}$, there are two possibilities:
	\begin{enumerate}[label=\ensuremath{(\arabic*)},leftmargin=*,itemsep=.4ex]
		\item there is some~$k\neq i$ with $k\in\pi(j)$ and $k\notin B$, or
		\item there is some~$k\neq i$ with $k\notin\pi(j)$ and $k\in B$.
	\end{enumerate}
	If~$(1)$, we have $\pi(j)=\pi(k)$ as well as $ik\hatsub=\bot$. As $jk\hatsub=jk$, it holds that $\pi\models ij\hatsub\wedge jk\hatsub$ but
	$\pi\not\models ik\hatsub$. If~$(2)$, however, we have $\pi(j)\neq\pi(k)$ and $ik\hatsub=\top$. As $jk\hatsub=jk$, it holds that $\pi\models ij\hatsub\wedge ik\hatsub$ but
	$\pi\not\models jk\hatsub$. In either case it follows that $\pi\not\models\trans\hatsub$.
	}
	\end{proof}
	The following example illustrates Lemma~\ref{lemma:forgetting1}.
	
	\begin{example}
		Consider the partition $\pi=[12\sep34\sep 5]$. Then, $\pi|_{-1}=\set{\set 2,\set{34},\set 5,\emptyset}$. 
		Let $1\vec\jmath=12,13,14,15$ be a fixed enumeration of $\vars_1$. Also let $\vec b_1=\bot,\top,\top,\bot$ and $\vec b_2=\bot,\top,\bot,\top$ be Boolean vectors (of length~$4$). Then, 
		\[
		[\,12\sep34\sep 5\,]\models\trans_{12,13,14,15\leftarrow\bot,\top,\top,\bot}\text.
		\]
		(This may be established, somewhat tediously, by painstakingly
	  checking all~$30$ conjuncts of the form $(kl\wedge lm)\to km$ of
	  $\trans$.)  Now, observe that $\set{j\midd 1j_{1\vec j\leftarrow\vec
	      b_1}}=\set{3,4}$ and that $\set{3,4}\in\pi_{-1}$.  On the other
	  hand, observe that $(13\wedge 15 \to 35)_{1\vec\jmath\leftarrow\vec
	    b_2}=(\top\wedge\top)\to 35$. It is easily established, however,
	  that $ [12|34|5]$ does not satisfy $(\top\wedge\top)\to
	  35$ and, hence, neither $\trans_{1\vec\jmath\leftarrow\vec
	    b_2}$. Finally, observe that $\set{j\midd 1j_{1\vec
	      \jmath\leftarrow\vec b_1}}=\set{3,5}$ and that~$\set{3,5}$ is not
	  in~$\pi|_{-1}$.
	\end{example}

	\bphnote{Lemma~\ref{lemma:forgetting2} is methinks the crucial for the characterisation of the stability concept. Silly me, only formulated Lemmas~\ref{lemma:forgetting1} and~\ref{lemma:forgetting2} for individual deviations, believing that similar results for ``group deviations'' would straightforwardly follow from them. It appears, however, we need ``group deviation'' variants so as also have (uniform) characterisations of blocking coalitions (which I believe can then be achieved in polynomial space) and core (which cannot). Obviously, these stronger versions would subsume Lemmas~\ref{lemma:forgetting1} and~\ref{lemma:forgetting2} and we do not need both.}

	\bphnote{For the moment the quantifiers $\hat\exists i$ and $\hat\forall i$ only play a role in the characterisation of Nash stability. }

	We now introduce the following abbreviation, where $i\vec\jmath=ij_1,\dots,ij_{n-1}$ is assumed to be a fixed enumeration of~$\vars_i$.
	\begin{align*}
	\hat\exists i\,\varphi\;	&	= \scalebox{1}[1]{$\displaystyle  \bigvee_{\vec b\in\set{\bot,\top}^{n-1}} (\varphi\wedge\trans)\hatsub$}
	\end{align*}
	Thus, $\hat\exists i$ can be understood as the operation
	of forgetting everything about player $i$ (in the sense of
	\cite{lin1994forget}) while taking the transitivity constraint into
	account. Intuitively,  $\hat\exists i\,\varphi$ signifies that given partition~$\pi$ player~$i$ can deviate to some coalition such that that~$\varphi$ is satisfied.


	%
	%

	\begin{proposition}\label{proposition:quantification}
		Let~$\pi$ be a partition,~$i$ a player, and~$\varphi$ a formula of~$L_N$. Then,
		\[
			\text{$\pi\models\hat\exists i\,\varphi	$	\text{\wordbox{iff}{iiiifffff} $\pi[i\to S]\models\varphi$ \;\wordbox[l]{for some $S\in \pi|_{-i}$,}{\quad for some $S\in \pi|_{-i}$}}}
		\]
	\end{proposition}

	\begin{proof}
		First assume $\pi\models\hat\exists i\,\varphi$. Then, $\pi\models(\varphi\wedge\trans)\hatsub$\quad for some $\vec b\in\set{\bot,\top}^{n-1}$. 
		Define $S=\set{j\midd ij\hatsub =\top}$. By Lemma~\ref{lemma:forgetting1}\ref{item:forgettting1ii}, we then obtain $\pi[i\to S]\models\varphi$.

		For the opposite direction, assume that $\pi[i\to S]\models\varphi$ for some $S\in\pi|_{-i}$. Define $\vec b=b_1,\dots,b_{n-1}$  as the Boolean vector of length~$n-1$ such that
		for every $1\le k\le n-1$,
		\[
			b_k	
			\mathwordbox{=}{==}
			\begin{cases}
				\top 	&	 \text{if $j\in S\cup\set i$}\\
				\bot	&	 \text{otherwise.}
			\end{cases}
		\]
		Then, clearly, $S=\set{j\midd ij\hatsub=\top}$. By Lemma~\ref{lemma:forgetting1}\ref{item:forgettting1ii}, it follows that $\pi\models\varphi\hatsub$. We may conclude that
		$\pi\models\hat\exists i\,\varphi$. 
	\end{proof}
	It is important to note, however, that the number of Boolean vectors of length~$k$ is exponential in~$k$. Accordingly, $\hat\exists i\,\varphi$  abbreviates a formula whose size is exponential in the size of~$\varphi$.

	\paragraph{Characterising group deviations}\label{subsection:characterising_group_diviations}

	\bphnote{The following lemma is useful for characterising both individual rationality and blocking coalitions. Additional text and an example should still be supplied.}

	Besides a single player deviating from its coalition and joining another, multiple players (from possibly different coalitions) could also deviate together and form a coalition of their own. This concept lies at the basis of, e.g., the core stability concept. We establish a formal connection between substitution and group deviations.

	Let~$T=\set{i_1,\dots,i_t}$ be a group of players.
	Observe that $|\vars_T|=\binom{n}{2}-\binom{n-t}{2}$ and let $\vec{\imath\jmath}_T$ be a fixed enumeration of $\vars_T$. By the \emph{$T$-separating Boolean vector} (given~$\vec{\imath\jmath}_T$) we define as the unique Boolean vector~$\vec b_T$ of length~$\binom{n}{2}-\binom{n-t}{2}$ such that for all~$i\in T$ and all~$j\in N$,
	\[
		ij_{\vec{\imath\jmath}_T\leftarrow\vec b_T}	\mathwordbox{=}{==}
		\begin{cases}
			\top &	\text{if $j\in T$,}\\
			\bot &	\text{otherwise.}
		\end{cases}
	\]
	Intuitively,~$\vec b_T$ represents the choice of group~$T$ to form a coalition of their own.
	Whenever~$T$ is clear from the context we omit the subscript in $\vec b_T$ and $\vec{\imath\jmath}_T$. The following characterisation now holds.

	\begin{lemma}\label{lemma:group_deviation}
		Let $(N,\gamma_1,\dots,\gamma_n)$ be a Boolean hedonic game,~$T$ a group of players,~$\pi$ a partition,~$\vec{\imath\jmath}$ a fixed enumeration of~$\vars_T$, and~$\vec b_T$ the corresponding $T$-separating Boolean vector. Then, for every formula~$\varphi\in L_N$,
		\[
		\text{$\pi\models\varphi_{\vec{\imath\jmath}\leftarrow\vec b_T} $ \wordbox{ if and only if }{ iif and only iff } $\pi[T\to\emptyset]\models\varphi$.}
		\]
	\end{lemma}

		\section{{Characterising Solutions}}
	  Our task in this section is to show how the various solution
	  concepts we introduced above can be \emph{characterised} as formulas
	  of our propositional language.  Let~$f$ be a function mapping each Boolean hedonic game~$G$ for~$N$ to a formula~$f(G)$ of~$L_N$. Given a solution concept~$\theta$,
	  we say that~$f$ is a {\em
	    characterisation} of~$\theta$ if for every Boolean hedonic game~$G$ on~$N$
	  and every partition $\pi$, we have that~$\pi$ is a solution
	  according to~$\theta$ for game~$G$ if and only if $\pi \models
	  f(G)$. If, furthermore, there exists a polynomial $p$ such that
	  $|f(G)| \leq p(|N|)$, then $f$ is a {\em polynomial
	    characterisation} of $\theta$.

	\hanote{This definitions above can feature more prominently perhaps in definition environment.}

	Once we have a characterisation of $\theta$, we know that there is a one-to-one correspondence between the partitions of $N$ satisfying $\theta$ and the models of $f(G)$. Therefore, given a Boolean hedonic game~$G$: 
	\begin{itemize}
	\item checking whether there exists a partition satisfying $\theta$ in~$G$ amounts to checking whether~$f(G)$ is satisfiable; 
	\item computing a partition satisfying~$\theta$ in~$G$ amounts to finding a model of~$f(G)$; 
	\item computing all partitions satisfying~$\theta$ in~$G$ amounts to finding all models of~$f(G)$.
	\end{itemize}

	Thus, once we have a characterisation of a solution concept, one can
	use a SAT solver to find (some or all) or to check the existence of partitions that satisfy it.
	  This carries over to {\em conjunctions} of solution
	concepts. For instance, if individual rationality is characterised by
	$f_{\mathit{IR}}$ and envy-freeness by $f_{\mathit{EF}}$, the there is
	a one-to-one correspondence between the individual rational envy-free
	partitions for~$G$ and the models of $f_{\mathit{IR}}(G) \wedge
	f_{\mathit{EF}}(G)$. More generally, these techniques can be used for finding or checking partitions satisfying~$\theta$ that also have certain other properties expressible in~$L_N$.

In the remainder of the section we focus on how a number of classical solution concepts, and see how they can be characterised in our logic. 

		\paragraph{Individual rationality, perfection, and  optimality}\label{ir}
	
		Recall that a partition is individually rational if any player is at
	  least as happy in her coalition as being alone, that is, no player
	  would prefer to leave her coalition to form a singleton
	  coalition. Now we have the following characterisation of individual rationality in our logic.

	\bphnote{The definition of the characterising formula could perhaps be presented in the main text, rendering the formulation of the proposition shorter and clearer.  \textbf{The important thing is the definition of the single Boolean vector $\vec b=(\bot,\dots,\bot)$}.}

		\begin{proposition}\label{proposition:individual_rationality}
		Let	$(N,\gamma_1,\dots,\gamma_n)$ be a Boolean hedonic game, let~$i$ be a player with goal~$\gamma$, and let $\pi$ be a partition. Let, furthermore, 
		$i\vec\jmath$ be a fixed enumeration of~$\vars_i$ and let
		$\vec b=\bot,\dots,\bot$ be the Boolean vector of length~$n-1$ only containing~$\bot$. Then,
			\begin{enumerate}[label=\ensuremath{(\roman*)},leftmargin=2.5em,itemsep=.5ex]
				\item  
			\text{$\pi$ is acceptable to~$i$}
			iff
			$\pi\models(\gamma_i)\hatsub\to \gamma_i$,
			\item
			\scalebox{1}[1]{$\pi$ is individually rational iff $\displaystyle\pi\models\bigwedge_{i\in N}\big((\gamma_i)\hatsub\to \gamma_i\big)$.}
	\end{enumerate}
	\end{proposition}

	\begin{proof}
		We only give the proof for~$(i)$, as~$(ii)$ follows as an immediate consequence. 
			{
		 For~$(i)$, merely consider the following equivalences, of which the third one follows from  Lemma~\ref{lemma:forgetting1}\ref{item:forgettting1ii}.
		\[
		\renewcommand{\arraystretch}{1.3}
			\begin{array}{lll}
					\multicolumn{1}{l}{\text{$\pi$ is acceptable to~$i$}}	
					&	\text{iff}	&	\text{$\pi\mathrel{R_i}\pi[i\to\emptyset]$}\\
					&	\text{iff}	&	\text{$\pi[i\to\emptyset]\models\gamma_i$ implies $\pi\models\gamma_i$}\\
					&	\text{iff}	&	\text{$\pi\models(\gamma_i)\hatsub$ implies $\pi\models\gamma_i$}\\				
					&	\text{iff}	&	\text{$\pi\models(\gamma_i)\hatsub\to \gamma_i$.}	
			\end{array}
		\]
		This concludes the proof.
		}
	\end{proof}
	To illustrate Proposition~\ref{proposition:individual_rationality} we consider again \exref{ex1}.
	{
	\setcounter{example}{0}
	\begin{example}[continued]
	In the game of our example, all partitions are acceptable to player~$1$, whose goal is given by $\gamma_1=123 \vee 124 \vee 134$.
 Let~$\vars_1$ be enumerated by $1\vec\jmath=12,13,14$ and let $\vec b=\bot,\bot,\bot$. Then, $(\gamma_2)_{12,13,14\leftarrow\bot,\bot,\bot}$ is $\TRANS$-equivalent to~$\bot$ and, hence, $\pi\models(\gamma_2)_{12,13,14\leftarrow\bot,\bot,\bot}\to\gamma_1$ for all partitions~$\pi$.
	According to Proposition~\ref{proposition:individual_rationality} this signifies that to player~$1$ every partition is acceptable.

	Now consider player~$4$, whose goal is given by $\neg423$, that is, by $\neg(42\wedge 43)$. Let~$\vars_4$ be enumerated by $41,42,43$ and let $\vec b=\bot,\bot,\bot$.
	Then, $\neg(42\wedge 43)_{41,42,43\leftarrow\bot,\bot,\bot}=\neg(\bot\wedge\bot)$, which is obviously $\TRANS$-equivalent to~$\top$.
	 Hence,
	\[
		\text{$\pi\models \neg(42\wedge 43)_{41,42,43\leftarrow\bot,\bot,\bot}$ \wordbox{if and only if}{if and only if} $\pi\models \neg(42\wedge 43)$,}
	\]
	meaning that a partition~$\pi$ is acceptable to  player~$4$ if and only if~$\pi$ satisfies his goal.
	\end{example}	
	}

%

		The logical characterisation of perfect perfect partition is immediate, as witnessed by the following proposition.

		\begin{proposition} Let	$(N,\gamma_1,\dots,\gamma_n)$ be a Boolean hedonic game. Then,
		a partition~$\pi$ is perfect \wordbox{if and only if}{ if and only if } $\displaystyle\pi \models \bigwedge_{i\in N} \gamma_i\text.$
		\end{proposition}
		As a consequence, a perfect partition exists if and
	  only if the formula $\mathit{trans}\wedge \bigwedge_{i\in N} \gamma_i$ is
	  satisfiable. Moreover, finding a social welfare maximising partition reduces to
	  finding valuation satisfying a maximum number of formulas~$\gamma_i\wedge\trans$, that is,
	  to solving a {\sc maxsat} problem.
	  
	  Leveraging the same idea of iteratively checking whether a perfect partition can be found for a subset of agents, one can compute Pareto optimal solutions for a given game. A subset~$\Psi$ of formulas is said to be a \emph{maximal trans-consistent} if both
	  \begin{enumerate}[label=$(\roman*)$] 
	  \item $\Psi\cup\set{\trans}$ is consistent, and 
	  \item $\Psi'\cup\set{\trans}$ is inconsistent for all sets of formulas~$\Psi'$ with~$\Psi\subsetneq\Psi'$.
	  \end{enumerate}
We now have the following proposition.
	  \begin{proposition}
	  A partition~$\pi$ of a Boolean hedonic game is Pareto optimal if and only if $\set{\gamma_i \midd \pi \models \gamma_i}$ is a maximal $trans$-consistent subset of $\{\gamma_1, \ldots, \gamma_n\}$
	  \end{proposition}
Algorithms for computing maximal consistent subsets are well-known and could thus be exploited for the computation of Pareto optimal partitions.
	  %




	\bphnote{Social welfare maximisation in the context of dichotomous or Boolean games seems to be equivalent to Pareto optimality. It seems to me that it is hard to characterise Pareto efficiency in a general way without ``quantifying'' over all partitions.}
	\jlnote{Yes, it is possible in polynomial space using cardinality formulas (which requires additional propositional symbols in the language). Probably we should not do this here, but in the future long version of the paper.} 

		\paragraph{Nash stability}

	Recall that a partition~$\pi$ is Nash stable, if no player~$i$ wishes to leave his coalition~$\pi(i)$ and join another (possibly empty) coalition so as to satisfy his goal. Leveraging our results from Section~\ref{section:substition_and_deviation}, we obtain the following characterisation of this fundamental solution concept.

	\begin{proposition}
		Let $(N,\gamma_1,\dots,\gamma_n)$ be a Boolean hedonic game and~$\pi$ a partition. Then,
		\[
			\text{$\pi$ is Nash stable}
			\wordbox{ if and only if }{ iif and only iff }
			\text{$\pi\models\bigwedge_{i\in N}\big((\hat\exists i\,\gamma_i)\to\gamma_i\big)$.}
		\] 
	\end{proposition}
	\begin{proof}
		Consider an arbitrary player~$i$ and observe that following equivalences hold. The fourth equivalence holds in virtue of Proposition~\ref{proposition:quantification}. The third one is a standard law of logic: merely observe that $\pi\models\gamma_i$ is not dependent on~$S$.
		\[
		\renewcommand{\arraystretch}{1.3}
			\begin{array}{l@{\;\;}l}
					\multicolumn{2}{l}{\text{$\pi$ is Nash stable}}	\\
					\text{iff}	&	\text{for all $i\in N$ and $S\in\pi|_{-i}$: $\pi\mathrel{R_i}\pi[i\to S]$}\\
					\text{iff}	&	\text{for all $i\in N$ and $S\in\pi|_{-i}$: if $\pi[i\to S]\models\gamma_i$ then $\pi\models\gamma_i$}							 	\\
					\text{iff}	&	\scalebox{1}[1]{\text{for all $i\in N$: if $\pi[i\to S]\models\gamma_i$ for some $S\in\pi|_{-i}$ then $\pi\models\gamma_i$}}							 	\\
					\text{iff}	&	\text{for all $i\in N$: if $\pi\models\hat\exists i\,\gamma_i$ then $\pi\models\gamma_i$}\\
					\text{iff}	&	\text{for all $i\in N$: $\pi\models(\hat\exists i\,\gamma_i)\to\gamma_i$}							 									\\	
					\text{iff}	&	\text{$\pi\models\bigwedge_{i\in N}\big((\hat\exists i\,\gamma_i)\to\gamma_i\big)$}							 												 							 	
			\end{array}
		\]
		This concludes the proof.
	\end{proof}
	Our running example illustrates this result.
	{
	\setcounter{example}{0}
	\begin{example}[continued]
		Consider again the game of Example~\ref{ex1}. Partition~$[123|4]$ satisfies each player's goal and, consequently, is Nash stable. We also have that
$
	[123|4]\models\gamma_1\wedge\gamma_2\wedge\gamma_3\wedge\gamma_4
$
and, thus,
\[
	[123|4]\models\bigwedge_{i\in N}\big((\hat\exists i\,\gamma_i)\to\gamma_i\big)\text.
\]
Now recall that for partition~$\pi=[1|23|4]$ player~$2$'s goal is not satisfied and that she cannot deviate and join another coalition to make this happen. 
		In this case, $\pi|_{-2}=\set{\set 1,\set 3,\set 4}$. Moreover, $\pi[2\to\set 1]=[12|3|4]$, $\pi[2\to\set 3]=[1|23|4]$, and $\pi[2\to\set 4]=[1|3|24]$. 
		Since,
	$
					[12|3|4] 	\not\models\gamma_2$,
	$				[1|23|4] 	\not\models\gamma_2$, and
	$				[1|3|24] 	\not\models\gamma_2$,	
		it follows that $\pi\not\models\hat\exists 2\,\gamma_2$. Hence, $\pi\models(\hat\exists 2\,\gamma_2)\to\gamma_2$.
		Player~$1$, however, could deviate from~$\pi_2$ and join~$\set{2,3}$ and thus have his goal satisfied. Thus, $\pi$ is not Nash stable.
		Now observe that $\set{2,3}\in\pi|_{-1}$ and that $\pi[1\to\set{2,3}]=[123|4]$. Moreover,
		$[123|4]\models\gamma_1$. As thus $\pi\models\hat\exists1\,\gamma_1$, also $\pi\not\models(\hat\exists1\,\gamma_1)\to\gamma_1$. We may
		conclude that 
		\[
			[1|23|4]\not\models\bigwedge_{i\in N}\big((\hat\exists i\,\gamma_i)\to\gamma_i\big)\text.
		\]
		 	\end{example}
		}
	Nash stable partitions are not guaranteed to exist in Boolean hedonic games. The two-player game $(\set{1,2},12,\neg 21)$ witnesses this fact, as can easily be appreciated. The translation into a SAT instance gives us a way to compute all Nash stable partitions of a given Boolean hedonic game. Recall, however, that the size of~$\hat\exists i\,\gamma_i$ is generally exponential in the size of~$\gamma_i$.

		\paragraph{Core and strict core stability}

		\bphnote{The above Lemma and Proposition seem to have been commented out...}

	Core and strict core stability relate to group deviations much in the same way as Nash stability relates to individual deviations.
Group deviations we characterised in Section~\ref{subsection:characterising_group_diviations}. 
We thus find that Lemma~\ref{lemma:group_deviation} yields a straightforward  characterisation in our logic of a specific group blocking or weakly blocking a given partition.

	\begin{proposition}\label{proposition:blocking_coalition}
	Let	$(N,\gamma_1,\dots,\gamma_n)$ be a Boolean hedonic game and~$T$ a group of players, and $\pi$ be a partition. Let, furthermore, 
	$\vec{\imath\jmath}$ a fixed enumeration of~$\vars_T$ and  
	$\vec b$ \emph{the corresponding $T$-separating Boolean vector}. Then,
	\begin{enumerate}[label=\ensuremath{(\roman*)},leftmargin=2.5em,itemsep=.5ex]
		\item \text{
			\text{$T$ blocks~$\pi$}
			\wordbox{ if and only if }{ iif and only iff }
			$\pi\models\displaystyle\bigwedge_{i\in T}\big(\neg\gamma_i\wedge(\gamma_i)\hatsubb\big)$,
		}
		\item $T$ weakly blocks~$\pi$ if and only if
	\[
		\pi\models\bigwedge_{j\in T}\big(\gamma_j\to(\gamma_j)\hatsubb\big)\wedge\bigvee_{i\in T}\big(\neg\gamma_i\wedge(\gamma_i)\hatsubb\big)\text.
	\]
	\end{enumerate}
	\end{proposition}

	\begin{proof} We give the proof for~$(i)$, as the one for~$(ii)$ runs along analogous lines.
	Consider the following equivalences, of which the third one follows immediately from Lemma~\ref{lemma:group_deviation}.
		\[
		\renewcommand{\arraystretch}{1.3}
			\begin{array}{lll}
					\multicolumn{1}{l}{\text{$T$ blocks $\pi$}}
					&	\text{iff}	&	\text{for all $i\in T$: $\pi[T\to\emptyset]\mathrel{P_i}\pi$}\\
					&	\text{iff}	&	\text{for all $i\in T$: $\pi[T\to\emptyset]\models\gamma_i$ and $\pi\not\models\gamma_i$}\\
					&	\text{iff}	&	\text{for all $i\in T$: $\pi\models(\gamma_i)\hatsubb$ and $\pi\not\models\gamma_i$}\\
					&	\text{iff}	&	\text{$\pi\models\displaystyle\bigwedge_{i\in T}\big(\neg\gamma_i\wedge(\gamma_i)\hatsubb\big)$.}\\
			\end{array}
		\]
		This concludes the proof.
	\end{proof}



	Observe that the size of $\bigwedge_{i\in T}\big(\neg\gamma_i\wedge(\gamma_i)\hatsubb\big)$  is obviously polynomial in $\sum_{i\in T}|\gamma_i|$ and, hence, a partition~$\pi$ being blocking by particular group~$T$ of players can be polynomially characterised.
	It might also be worth observing that this  characterisation is  reminiscent of that for individual rationality and, surprisingly, much more so than of the one for Nash stability.

	As a corollary of Proposition~\ref{proposition:blocking_coalition} and de Morgan laws, we  obtain the following characterisations of a partition being core stable and of a partition being strict core stable. The characterisations, however, involve a conjunctions over all groups of players and as such is not polynomial.

	\begin{corollary}
	Let	$(N,\gamma_1,\dots,\gamma_n)$ be a Boolean hedonic game and~$\pi$ be a partition. Let for each coalition~$T$, $\vec{\imath\jmath}$ be an enumeration of~$\vars_T$ and~$\vec b$ the corresponding $T$-separating Boolean vector.  Then, 
	\begin{enumerate}[label=\ensuremath{(\roman*)},leftmargin=2.5em,itemsep=.5ex]
	\item
	$\pi$ is core stable if and only if 
	\scalebox{1}[1]{$
		\displaystyle\pi\models \bigwedge_{T\subseteq N}\bigvee_{i\in T}\big((\gamma_i)_{\vec{\imath\jmath}\leftarrow\vec b}\to \gamma_i\big)
	$,}
	\item Then, $\pi$ is strict core stable if and only if 
	\[
		\pi\models \bigwedge_{T\subseteq N}\big(\bigvee_{j\in T}(\gamma_j\wedge\neg(\gamma_j)\hatsubb)\vee\bigwedge_{i\in T}((\gamma_i)\hatsubb\to\gamma_i)\big)\text.
	\]	
	\end{enumerate}
	\end{corollary}

	%
	%

	Although core stable coalition structure are not guaranteed to exist in general hedonic games, the restriction to dichotomous preferences allows us to derive this positive result.

	\begin{proposition}
	For every Boolean hedonic game, a core stable coalition structure is guaranteed to exist.
	\end{proposition}
	\begin{proof}

		We initialise $N'$ to $N$ and partition $\pi$ to $\{\emptyset\}$. We find a maximal subset of $S\subset N'$ for which all players are in an approved coalition that satisfies their formulas. We modify $\pi$ to $\pi\cup \{S\}$ and $N'$ to $N'\setminus S$. The procedure is repeated until no such maximal subset $S$ exists. If $N'\neq \emptyset$, then $\pi$ is set to $\pi\cup \{\{i\}\midd i\in N'\}$. 
	
		We now argue that $\pi$ is core stable.
	We note that each player who was in some subset $S$ will never be part of a blocking coalition. If $N'$ was non-empty in the last iteration, then no subset of players in $N'$ can form a deviating coalition among themselves.
	\end{proof}
	
By contrast, a strict core stable partition is not guaranteed to exist. To see this consider the three-player Boolean hedonic game $(\set{1,2,3},12,21\vee 23,32)$. It is not hard to see that each of the five possible partitions is weakly blocked by either $\set{1,2}$ or $\set{2,3}$.


	\paragraph{Envy-freeness}

	Recall that a partition is envy-free if no player would strictly prefer to exchange places with another player. Observe that for the trivial partitions 
$\pi^0=[1\sep\cdots\sep n]$ and	
$\pi^1=[1,\dots,n]$,	
	 we have $\pi^0[i\leftrightarrows j]=\pi^0$ and $\pi^1[i\leftrightarrows j]=\pi^1$ for all players~$i$ and~$j$. Accordingly~$\pi^0$ and~$\pi^1$ are envy-free. Envy-free partitions are thus guaranteed to exist in our setting.  The following lemma allows us to derive a polynomial characterisation of envy-freeness.

	\bphnote{What follows implies a \textbf{polynomial} characterisation of envy-freeness. It is important to observe that the characterisation depends on the hedonic aspect, that is, we are making essential use of the fact that each player~$i$'s goal is a formulas of~$L_i$ rather than~$L_N$.}

	\bphnote{How to get rid of the bloody case distinctions in the following proof?!}

	\newcommand{\swapsub}{_{i\vec k,j\vec k\leftarrow j\vec k,i\vec k}}

	\begin{lemma}\label{lemma:swapsub}
		Let $(N,\gamma_1,\dots,\gamma_n)$ be a Boolean hedonic game and~$i$ and~$j$ players in~$N$, and $\varphi$ a formula in~$L_N$. Fix, furthermore, an enumeration
		$k_1,\dots,k_{n-2}$ of $N\setminus\set{i,j}$ and let ${i\vec k}=ik_1,\dots,ik_{n-2}$ and $j\vec k=jk_1,\dots,jk_{n-2}$ enumerate $\vars_i\setminus\set{ij}$ and $\vars_j\setminus\set{ji}$, respectively. Then,
		\[
			\text{
				$\pi\models\varphi\swapsub$
			\wordbox{ if and only if }{ if and only if }
				$\pi[i\leftrightarrows j]\models\varphi$.
			}
		\]
	\end{lemma}

	\begin{proof}
		With~$i\vec k$ and~$j\vec k$ being fixed we write~$\varphi'$ for~$\varphi\swapsub$. 		The proof is then by induction on~$\varphi$. 
		{
		\renewcommand{\swapsub}{'}
	
		For the basis, let $\varphi=lm$. There are three possibilities:
		\[
		\scalebox{1}[1]{
\begin{enumerate*}[label=$(\alph{*})$,leftmargin=2.25em]
			\item $lm=ij$,						
			\item $lm\in (V_i\cup V_j)\setminus\set{ij}$, and
			\item $lm\notin V_i\cup V_j$.
		\end{enumerate*}}
		\]
		If~$(a)$, we have that $lm\swapsub=ij\swapsub=ij=lm$.
		Now, either $\pi(i)=\pi(j)$ or $\pi(i)\neq\pi(j)$. If the former, $\pi[i\leftrightarrows j]=\pi$ as well as both $\pi\models ij\swapsub$ and $\pi[i\leftrightarrows j]\models ij$. If the latter, however, it can easily be seen that both $\pi\not\models ij\swapsub$ and $\pi[i\leftrightarrows j]\not\models ij$.

For case~$(b)$, we may assume without loss of generality that $lm=ik$ for some $k\neq j$. Then, $ik\swapsub = jk$. In case $\pi(i)=\pi(j)$, obviously, $\pi=\pi[i\leftrightarrows j]$ as well as $k\in \pi(i)$ if and only if $k\in\pi(j)$. Hence, $\pi\models ik\swapsub$ if and only if $\pi[i
	\leftrightarrows j]\models ik$. So, assume $\pi(i)\neq \pi(j)$. Now, either 
	\begin{enumerate*}[label=$(\roman{*})$,leftmargin=2.25em]
	\item $k\in\pi(i)$ and $k\notin\pi(j)$, 
	\item $k\notin\pi(k)$ and $k\in\pi(j)$, or 
	\item $k\notin\pi(i)$ and $k\notin\pi(j)$.
	\end{enumerate*} 
	If~$(i)$, $\pi\models ik\swapsub$ as well as $\pi[i\leftrightarrows j]\models jk$. In cases~$(ii)$ and~$(iii)$, we have $\pi\not\models ik\swapsub$ and $\pi[i\leftrightarrows j]\not\models jk$.
	
	Finally, if~$(c)$, we have $lm\swapsub=lm$. As $l,m\notin\set{i,j}$, it can then easily be seen that $\pi\models lm\swapsub$ if and only if $\pi[i\leftrightarrows j]\models lm$.

The cases $\varphi=\neg\psi$ and $\varphi=\psi\to\chi$ follow by induction.
%
		}
	\end{proof}
We are now in a position to state the following result.

	\begin{proposition}\label{proposition:envyfree}
		Let $(N,\gamma_1,\dots,\gamma_n)$ be a Boolean hedonic game. Furthermore, for every two players,~$i$ and~$j$, and enumeration
		$k_1,\dots,k_{n-2}$ of $N\setminus\set{i,j}$, let ${i\vec k}=ik_1,\dots,ik_{n-2}$ and $j\vec k=jk_1,\dots,jk_{n-2}$ enumerate $\vars_i\setminus\set{ij}$ and $\vars_j\setminus\set{ij}$, respectively.
	Then,
		\[
			\scalebox{1}[1]{
				$\pi$ is envy-free \wordbox{if and only if}{ if and only if } $\displaystyle\pi\models\bigwedge_{i,j\in N}\big((\gamma_i)\swapsub\to\gamma_i\big)$.
			}
		\]
	\end{proposition}
	
	\begin{proof}
	By virtue of Lemma~\ref{lemma:swapsub}, the following equivalences hold:
		\[
		\renewcommand{\arraystretch}{1.3}
			\begin{array}{ll}
					\multicolumn{2}{l}{\text{$\pi$ is envy-free}}\\
					\text{iff}	&	\text{for all $i,j\in N$: $\pi\mathrel{R_i}\pi[i\leftrightarrows j]$}\\
					\text{iff}	&	\text{for all $i,j\in N$: $\pi[i\leftrightarrows j]\models\gamma_i$ implies $\pi\models\gamma_i$}\\
					\text{iff}	&	\text{for all $i,j\in N$: $\pi\models(\gamma_i)\swapsub$ implies $\pi\models\gamma_i$}\\
					\text{iff}	&	\text{for all $i,j\in N$: $\pi\models(\gamma_i)\swapsub\to\gamma_i$}\\
					\text{iff}	&	\text{$\pi\models\displaystyle\bigwedge_{i,j\in N}\big((\gamma_i)\swapsub\to\gamma_i\big)$}
			\end{array}
		\]
		This concludes the proof.
	\end{proof}
	Observe that the size of $\bigwedge_{i,j\in N}\big((\gamma_i)\swapsub\to\gamma_i\big)$  is clearly polynomial in $\sum_{i\in T}|\gamma_i|$. Hence, a partition~$\pi$ being envy-free can be polynomially characterised.

	{
	\setcounter{example}{0}
	\begin{example}[continued]
	Recall that $\gamma_3=(31\vee 32)\wedge\neg 34$ and that player~$3$ envies player~$4$ if partition~$\pi'=[1|24|3]$ obtains. 
	To see how this is reflected by Proposition~\ref{proposition:envyfree}, let~$31,32$ and~$41,42$ enumerate $V_3\setminus\set{34}$ and $V_4\setminus\set{43}$, respectively. Then,
	\[
			((31\vee 32)\wedge\neg 34)_{31,32,41,42\leftarrow 41,42,31,32}
			=
			(41\vee 42)\wedge\neg 34\text.
	\]
	Now, both $\pi'\models (41\vee 42)\wedge\neg 34$ and $\pi'\not\models(31\vee 32)\wedge\neg 34$, and, hence,
$
		\pi'\not\models(\gamma_3)_{34,31,32\leftarrow43,41,42}\to\gamma_3
$.
		\end{example}
		}


	\mjwnote{
	\section{The Scope of Boolean Hedonic Games}

	\subsection{Boolean Hedonic Games Unleashed}

	Some informal remarks concerning
	\begin{itemize}
		\item the possibility of extending the framework with general preferences---there are several possibilities,
		\item the possibility of extending the framework to general coalition formation games,
		\item the completeness of Boolean hedonic games, in the sense that every dichotomous hedonic game can be represented by a Boolean hedonic game.
	\end{itemize}

	\subsection{Applications}

	\bphnote{J\'er\^ome has some ideas for this section.}

	\begin{itemize}
		\item Kidney exchange
		\item Marriage markets
	\end{itemize}
	}

	\section{Related Work and Conclusions}

	Our motivation and approach is strongly reminiscent of the setting of
	Boolean games in the context of non-cooperative game
	theory~\citep{HHMW01a}.  A major difference with Boolean games and
	propositional hedonic games is that in Boolean games, players have
	preferences over outcomes, where an outcome is a truth assignment to
	outcome variables, and each outcome variable is controlled by a
	specific player.  This control assignment function, which is a central
	notion in Boolean games, has no counterpart here, where the outcome is
	a partition of the players. However, there are technical
	similarities with and conceptual connections to Boolean games, especially when characterising solution
	concepts. For instance, the characterisation of Nash stable partitions
	by propositional formulas (Section 4) is similar to the
	characterisation of Nash equilibria by propositional formulas in
	Boolean games as by \citet{bonzon:2009a}.  The basic Boolean games model
	of~\citet{HHMW01a} was adapted to the setting of cooperative games
	by \citet{dunne:2008a}. However, the logic used to specific player's
	goals in the work of~\citeauthor{dunne:2008a} was not intended for
	specifying desirable coalition structures, as we have done in the
	present paper.  
	
	Our work also shares some common ground with the work
	of \citet{bonzon2012effectivity}, who study the formation of efficient
	coalitions in Boolean games, that is, coalitions whose joint abilities
	allow their members to jointly achieve their goals. Our work also
	bears some resemblance to the work of \citet{ElWo09a}, who were
	interested in using logic as a foundation upon which to build a
	compact representation scheme for hedonic games; more precisely, their
	work made use of weighted Boolean formulas, and was inspired by the
	\emph{marginal contribution nets} representation for cooperative games
	in characteristic function form 
	proposed by\citet{ieong:2005a}. The focus of~\citet{ElWo09a}, however, was more
	on complexity issues than in finding exact characterisations for
	solution concepts.

	Finally, our work 
	contributes to the extensive literature on compact
	representations for cooperative games, which has expanded rapidly over
	the past decade~\citep{chalkiadakis:2011a}.

	Our characterisations of solution concepts enable to compute, using an off-the-shelf
	SAT solver, a partition or all partitions satisfying a solution concept or a logical combination of solution concepts.
	Of course, this translation is interesting only when we cannot do better. For instance, for solution concepts
	leading to a polynomial characterisation, we cannot do better if and only if the corresponding decision problem
	is {\sf NP}-complete. Identifying the complexity of finding partitions satisfying solution concepts for Boolean hedonic games
	is therefore the most immediate direction of further research.

\balance
	There are at least three more directions in which our work might be further developed.  
	First, we could think of relaxing our restriction to
	dichotomous preferences and study more general hedonic games with
	compact logical representations and derive exact
	characterisations of solution concepts. There are several ways in which more general preferences can be incorporated in our logical framework for hedonic games. For instance, instead of a single goal, we could associate with each player a prioritised set of goals. The different possibilities in this respect, however, vary in their level of sophistication. For some of the cruder extensions our results extend naturally and straightforwardly. For the more sophisticated settings more research seems to be required, which falls beyond the scope of this paper.  
	
	Second, our restriction to
	{\em hedonic} preferences can also be relaxed, so that players may
	have preferences that do depend not only on on the coalition to which
	they belong. This would also pave the way to a more general {\em logic
	  of coalition structures}. Solution concepts, once generalised, can
	hopefully be characterised. (We have positive preliminary results that
	go into this direction). 
	
	A third topic of future research would be the characterisation of classes of hedonic and coalition formation games in our logic. As mentioned above, various classes of hedonic games that allow for a concise representation have been proposed in the literature. It would be interesting to see whether these classes can also be polynomially characterised in our logic.



	\paragraph{Acknowledgments}
The authors would like to thank the anonymous referees of LOFT 2014 for their constructive comments.
Haris Aziz has been supported by NICTA which is funded by the Australian
Government as represented by the Department of Broadband,
Communications and the Digital Economy and the Australian Research
Council through the ICT Centre of Excellence program. 
J{\'e}r{\^o}me Lang has been supported by the ANR project CoCoRICo-CoDec.
Paul Harrenstein and Michael Wooldridge have been supported by the ERC under Advanced Grant 291528 (``RACE'').


\begin{thebibliography}{19}
\providecommand{\natexlab}[1]{#1}
\providecommand{\url}[1]{\texttt{#1}}
\expandafter\ifx\csname urlstyle\endcsname\relax
  \providecommand{\doi}[1]{doi: #1}\else
  \providecommand{\doi}{doi: \begingroup \urlstyle{rm}\Url}\fi

\bibitem[Aziz et~al.(2013)Aziz, Brandt, and Harrenstein]{ABH11c}
H.~Aziz, F.~Brandt, and P.~Harrenstein.
\newblock Pareto optimality in coalition formation.
\newblock \emph{Games and Economic Behavior}, 82:\penalty0 562--581, 2013.

\bibitem[Bogomolnaia and Jackson(2002)]{bogomolnaia:2002}
A.~Bogomolnaia and M.~O. Jackson.
\newblock The stability of hedonic coalition structures.
\newblock \emph{Games and Economic Behaviour}, 38:\penalty0 201--230, 2002.

\bibitem[Bogomolnaia and Moulin(2004)]{BoMo04a}
A.~Bogomolnaia and H.~Moulin.
\newblock Random matching under dichotomous preferences.
\newblock \emph{Econometrica}, 72\penalty0 (1):\penalty0 257--279, 2004.

\bibitem[Bogomolnaia et~al.(2005)Bogomolnaia, Moulin, and Stong]{BMS05a}
A.~Bogomolnaia, H.~Moulin, and R.~Stong.
\newblock Collective choice under dichotomous preferences.
\newblock \emph{Journal of Economic Theory}, 122\penalty0 (2):\penalty0
  165--184, 2005.

\bibitem[Bonzon et~al.(2009)Bonzon, Lagasquie-Schiex, Lang, and
  Zanuttini]{bonzon:2009a}
E.~Bonzon, M.-C. Lagasquie-Schiex, J.~Lang, and B.~Zanuttini.
\newblock Compact preference representation and {B}oolean games.
\newblock \emph{Autonomous Agents and Multi-Agent Systems}, 18\penalty0
  (1):\penalty0 1--35, 2009.

\bibitem[Bonzon et~al.(2012)Bonzon, Lagasquie-Schiex, and
  Lang]{bonzon2012effectivity}
E.~Bonzon, M.-C. Lagasquie-Schiex, and J.~Lang.
\newblock Effectivity functions and efficient coalitions in {Boolean games}.
\newblock \emph{Synthese}, 187\penalty0 (1):\penalty0 73--103, 2012.

\bibitem[Bouveret and Lang(2008)]{Bouveret08Jair}
S.~Bouveret and J.~Lang.
\newblock Efficiency and envy-freeness in fair division of indivisible goods:
  Logical representation and complexity.
\newblock \emph{Journal of AI Research}, 32:\penalty0 525--564, 2008.

\bibitem[Brams and Fishburn(2007)]{BrFi07c}
S.~J. Brams and P.~C. Fishburn.
\newblock \emph{Approval Voting}.
\newblock Springer, 2007.

\bibitem[Cechl{\'a}rov{\'a}(2008)]{Cech08a}
K.~Cechl{\'a}rov{\'a}.
\newblock Stable partition problem.
\newblock In \emph{Encyclopedia of Algorithms}, pages 885--888. Springer, 2008.

\bibitem[Cechl{\'a}rov{\'a} and Hajdukov{\'a}(2004)]{CeHa04a}
K.~Cechl{\'a}rov{\'a} and J.~Hajdukov{\'a}.
\newblock Stable partitions with {$\mathcal{W}$}-preferences.
\newblock \emph{Discrete Applied Mathematics}, 138\penalty0 (3):\penalty0
  333--347, 2004.

\bibitem[Chalkiadakis et~al.(2011)Chalkiadakis, Elkind, and
  Wooldridge]{chalkiadakis:2011a}
G.~Chalkiadakis, E.~Elkind, and M.~Wooldridge.
\newblock \emph{Computational Aspects of Cooperative Game Theory}.
\newblock Morgan-Claypool, 2011.

\bibitem[Dr{\`e}ze and Greenberg(1980)]{dreze:80a}
J.~H. Dr{\`e}ze and J.~Greenberg.
\newblock Hedonic coalitions: Optimality and stability.
\newblock \emph{Econometrica}, 48\penalty0 (4):\penalty0 987--1003, 1980.

\bibitem[Dunne et~al.(2008)Dunne, Kraus, {van der Hoek}, and
  Wooldridge]{dunne:2008a}
P.~E. Dunne, S.~Kraus, W.~{van der Hoek}, and M.~Wooldridge.
\newblock Cooperative {B}oolean games.
\newblock In \emph{Proceedings of the Seventh International Joint Conference on
  Autonomous Agents and Multiagent Systems (AAMAS-2008)}, pages 1015--1022,
  2008.

\bibitem[Elkind and Wooldridge(2009)]{ElWo09a}
E.~Elkind and M.~Wooldridge.
\newblock Hedonic coalition nets.
\newblock In \emph{Proceedings of the Eigth International Joint Conference on
  Autonomous Agents and Multiagent Systems (AAMAS-2009)}, pages 417--424, 2009.

\bibitem[Hajdukov{\'a}(2006)]{hajdukova:2006}
J.~Hajdukov{\'a}.
\newblock Coalition formation games: {A} survey.
\newblock \emph{International Game Theory Review}, 8\penalty0 (4):\penalty0
  613--641, 2006.

\bibitem[Harrenstein et~al.(2001)Harrenstein, {van der Hoek}, Meyer, and
  Witteveen]{HHMW01a}
P.~Harrenstein, W.~{van der Hoek}, J.-J. Meyer, and C.~Witteveen.
\newblock Boolean games.
\newblock In J.~{van Benthem}, editor, \emph{Proceedings of the 8th Conference
  on Theoretical Aspects of Rationality and Knowledge (TARK)}, pages 287--298,
  2001.

\bibitem[Ieong and Shoham(2005)]{ieong:2005a}
S.~Ieong and Y.~Shoham.
\newblock Marginal contribution nets: A compact representation scheme for
  coalitional games.
\newblock In \emph{Proceedings of the Sixth ACM Conference on Electronic
  Commerce (EC'05)}, Vancouver, Canada, 2005.

\bibitem[Konieczny and Pino-P\'{e}rez(2002)]{KP02}
S.~Konieczny and R.~Pino-P\'{e}rez.
\newblock Merging information under constraints: a logical framework.
\newblock \emph{Journal of Logic and Computation}, 12\penalty0 (5):\penalty0
  773--808, 2002.

\bibitem[Lin and Reiter(1994)]{lin1994forget}
F.~Lin and R.~Reiter.
\newblock Forget it!
\newblock In \emph{Working Notes of AAAI Fall Symposium on Relevance}, pages
  154--159, 1994.

\end{thebibliography}
\end{document}